\documentclass[pra,aps,floatfix,superscriptaddress,notitlepage,twocolumn,a4paper]{revtex4-1}
\usepackage[utf8]{inputenc}

\usepackage{fullpage} 
\usepackage{tikz} 
\usepackage{amsmath,amsfonts,amsthm, mathtools, mathrsfs}
\usepackage{float}
\usepackage{dsfont}
\usepackage{csquotes}
\usepackage{amssymb}
\usepackage{verbatim}
\usepackage{enumerate}
\usepackage{centernot}
\usepackage{subfigure}

\usepackage{hyperref}
\definecolor{darkred}  {rgb}{0.5,0,0}
\definecolor{darkblue} {rgb}{0,0,0.5}
\definecolor{darkgreen}{rgb}{0,0.5,0}
\hypersetup{
	colorlinks = true,
	urlcolor  = blue,         
	linkcolor = darkblue,     
	citecolor = darkgreen,    
	filecolor = darkred       
}

\usepackage{graphicx}
\usepackage[leftcaption]{sidecap}

\theoremstyle{definition}

\newtheorem{lemma}{Lemma}

\newcommand{\N}{\mathcal{N}}

\newcommand{\fid}{\mathsf{F}}

\definecolor{cool_green}{rgb}{0.0, 0.5, 0.0}
\definecolor{cool_blue}{rgb}{0.0, 0.0, 0.5}

\newcommand{\Tr}[1]{\operatorname{Tr}\!\left[#1\right]}
\newcommand{\PTr}[2]{\operatorname{Tr}_{{#1}}\!\left\{#2\right\}}
\def\>{\rangle}
\def\<{\langle}
\def\N#1{\left|\!\left|{#1}\right|\!\right|}

\def\mE{\mathcal{E}}

\def\openone{\mathds{1}}

\renewcommand{\qedsymbol}{\nobreak \ifvmode \relax \else
	\ifdim \lastskip<1.5em \hskip-\lastskip \hskip1.5em plus0em
	minus0.5em \fi \nobreak \vrule height0.75em width0.5em
	depth0.25em\fi}

\renewcommand{\ge}{\geqslant}
\renewcommand{\le}{\leqslant}

\begin{document}

\title{Thermodynamic Reverse Bounds\\for General Open Quantum Processes}

\author{Francesco Buscemi}
\email{buscemi@i.nagoya-u.ac.jp}
\affiliation{Department of Mathematical Informatics, Graduate  School  of  Informatics,  Nagoya University, Chikusa-ku, 464-8601 Nagoya, Japan}
\author{Daichi Fujiwara}
\affiliation{Department of Mathematical Informatics, Graduate  School  of  Informatics,  Nagoya University, Chikusa-ku, 464-8601 Nagoya, Japan}
\author{Naoki Mitsui}
\affiliation{Department of Mathematical Informatics, Graduate  School  of  Informatics,  Nagoya University, Chikusa-ku, 464-8601 Nagoya, Japan}
\author{Marcello Rotondo\thanks{asdasd}}
\affiliation{Department of Physics, Graduate School of Science, Nagoya University, Chikusa-ku, Nagoya 464-8602, Japan (until March 2019)}
\date{\today}

\begin{abstract}
	
	Various quantum thermodynamic bounds are shown to stem from a single tighter and more general inequality, consequence of the operator concavity of the logarithmic function. Such an inequality, which we call the ``thermodynamic reverse bound'', is compactly expressed as a quantum relative entropy, from which it inherits mathematical properties and meaning. As concrete examples, we apply our bound to evaluate the thermodynamic length for open processes, the heat exchange in erasure processes, and the maximal energy outflow in general quantum evolutions.
	
\end{abstract}

\maketitle

\section{Introduction}

Umegaki's quantum extension~\cite{umegaki1961} of Kullback--Leibler's relative entropy~\cite{kullback1951} is of vital importance for a rigorous formalization of quantum thermodynamics. In particular, the fact that quantum relative entropy, notwithstanding the noncommutativity of the underlying algebra of operators, shares many properties with its classical counterpart---most notably the monotonicity under coarse-grainings~\cite{lindblad1974}---is enough by itself to serve as a powerful tool for investigation.

In this paper we show how a much weaker property than monotonicity, namely, the operator concavity of the logarithmic function~\cite{Carlen2010} can be invoked to improve on various thermodynamic bounds for quantum nonequilibrium processes, possibly driven and open. {Besides the mere quantitative improvements over previously known bounds, our formula possesses the unique advantage of being expressed as a quantum relative entropy, thus inheriting all the mathematical properties, the conceptual advantages, and the physical relevance of one among the most prolific quantities in statistical sciences~\cite{jaynes_2003,cover2012elements,amari_2016}.} 

The paper proceeds as follows. After presenting the general mathematical bound, we show its physical relevance by deriving: (1), a new formula for the thermodynamic length, recovering that for closed processes as a special case, but directly applicable to open processes; (2), a tighter bound for the heat exchange in erasure processes; and, (3), a stronger formulation of the quantum second law as presented in~\cite{Tasaki2015} as a generalization of~\cite{Lenard1978} to general quantum open processes.

\section{The thermodynamic reverse bound}

We consider a system $S$ interacting with an ancilla $A$ from time $t=0$ to time $t=\tau>0$. We assume that the system and the ancilla are initially in a factorized state $\rho_0^{SA}=\rho^S_0\otimes\rho^A_0$. On the other hand, we make no particular assumptions at this point on the microscopic model (i.e., the Hamiltonian operator) underlying the interaction. We only need to know that the joint system-ancilla evolution from $t=0$ to $t=\tau$ is described by some unitary operator $U^{SA}$, so that the reduced system's dynamics is described by a completely positive trace-preserving (CPTP) linear map $\Phi$ as follows~\cite{Breuer-Petruccione}:
\begin{align}
\rho^S_\tau&=\PTr{A}{U^{SA}\ (\rho^S_0\otimes\rho^A_0)\ (U^{SA})^\dag}\label{eq:micro-model}\\
&=:\Phi(\rho^S_0)\;.
\end{align}

In order to discuss the thermodynamics of system $S$, we also introduce the system's Hamiltonian operators at initial and final times, namely,  $H^S_{ \{0,\tau\} }$, and the corresponding states at equilibrium at some inverse temperature $\beta>0$, namely, $\Gamma^S_{\{0,\tau\}}=\exp[\beta(F_{\{0,\tau\}}-H^S_{ \{0,\tau\} })]$, where $F_{\{0,\tau\}}:=-\beta^{-1}\ln\Tr{\exp(-\beta H^S_{ \{0,\tau\}})}$ denote the system's free energies at thermal equilibrium. In this paper we focus in particular on the difference between the system's average internal energy at final and initial times, namely,
\begin{equation}
\Delta E:=\Tr{\rho^S_\tau\ H^S_\tau}-\Tr{\rho^S_0\ H^S_0}\;.
\end{equation}
The above quantity, though always well-defined, is most meaningful when any interaction terms can be neglected in the energy computation. This is the case, for example, if at times $t=0$ and $t=\tau$ the interaction terms in the total Hamiltonian are negligible with respect to the system's internal energy. As a concrete example, one can think of a system that ``flies'' through a relatively small interaction region, where energy can be exchanged with the ancilla, but at times $t=0$ and $t=\tau$ is sufficiently far from such a region so that interactions can be neglected.

Denoting by $\Phi^\dagger$ the trace-dual map associated with $\Phi$, that is, the map satisfying $\Tr{\Phi(X)\ Y}=\Tr{X\ \Phi^\dagger(Y)}$ for all operators $X,Y$, by standard manipulations it is possible to show that
\begin{align}
&\beta(\Delta E-\Delta F)+D(\rho^S_0\|\Gamma^S_0)\nonumber\\
&=D(\rho^S_0\|\Phi^\dag(\Gamma^S_\tau))+\Tr{\rho^S_0\ L}\label{eq:main-eq}\\
&\ge D(\rho^S_0\|\Phi^\dag(\Gamma^S_\tau))\label{eq:main-ineq}\;,
\end{align}
where:
\begin{enumerate}[(i)]
	\item $\Delta F$ denotes the difference of equilibrium free energies: $\Delta F:=F_\tau-F_0$;
	\item $D(X\|Y)$, for $X,Y\ge 0$, $X\neq 0$, denotes an extension of Umegaki's quantum relative entropy to positive semi-definite operators (see, e.g.,~\cite{MDSFT2013}), namely, $D(X\|Y):=\Tr{X}^{-1}\Tr{X\ln X-X\ln Y}$, if $\operatorname{supp}X \subseteq \operatorname{supp}Y$, or $+\infty$ otherwise;
	\item $L$ denotes the operator defined by $L:=\ln[\Phi^\dag(\Gamma^S_\tau)] -\Phi^\dag(\ln\Gamma^S_\tau)$.
\end{enumerate}
In this paper we refer to inequality~(\ref{eq:main-ineq}) as the \textit{thermodynamic reverse bound}, because the trace-dual map $\Phi^\dag$, which for isolated processes coincides with the reverse evolution, is used to ``pull'' the final thermal equilibrium state from $t=\tau$ back to $t=0$.
	
The thermodynamic reverse bound holds because the operator $L$ in~\eqref{eq:main-eq} is always positive semidefinite. This fact is a consequence of the operator concavity of the logarithmic function (a result sometimes referred to as the L\"owner--Heinz Theorem; see, e.g., Theorem~2.6 in~\cite{Carlen2010}) together with the so-called Schwarz inequality for positive linear maps~\cite{Davis1961,Choi1974}. We notice however that, due to the fact that the trace-dual map of a CPTP map is CP but not necessarily TP, the operator $\Phi^\dag(\Gamma^S_\tau)$ appearing in~(\ref{eq:main-eq}) and~(\ref{eq:main-ineq}) is positive semidefinite but its trace may differ from one. As a consequence, the term $D(\rho^S_0\|\Phi^\dag(\Gamma^S_\tau))$ may be smaller than zero. A special case of~(\ref{eq:main-ineq}) is discussed from an information-theoretic perspective in~\cite{Buscemi-Das-Wilde}.

The above relations, Eqs.~(\ref{eq:main-eq}) and~(\ref{eq:main-ineq}), are to be compared with the following ones, that again can be easily verified by direct inspection, that is,
\begin{align}
&\beta(\Delta E-\Delta F)+D(\rho^S_0\|\Gamma^S_0)\nonumber\\
&=D(\Phi(\rho^S_0)\|\Gamma^S_\tau)+\Delta S\label{eq:standard-eq}\\
&\ge \Delta S\label{eq:standard-ineq}\;,
\end{align}
where $\Delta S:=-\Tr{\rho^S_\tau\ln\rho^S_\tau}+\Tr{\rho^S_0\ln\rho^S_0}$ is the entropy increase of the system, which can be positive or negative, depending on the initial state and the CPTP map $\Phi$. Inequality~(\ref{eq:standard-ineq}) follows because both $\Phi(\rho^S_0)=\rho^S_\tau$ and $\Gamma^S_\tau$ are two normalized density matrices, so that their relative entropy is non-negative.

It is interesting to notice that both relations~(\ref{eq:main-eq}) and~(\ref{eq:standard-eq}) only involve terms which depend solely on the system's reduced dynamics, and that both relations are expressed as the sum of two quantities, one of which is always non-negative. However, only the thermodynamic reverse bound~(\ref{eq:main-ineq}) is expressed in the form of a quantum relative entropy, thus carrying the useful interpretation of information-theoretic divergence~\cite{amari_2016}. In what follows, we will discuss this fact in connection with some situations of physical interest.

\subsection{Conditions for equality}

Before proceeding with the applications, we first comment on the conditions that imply equality in~(\ref{eq:main-ineq}), for all initial states $\rho_0^S$. More precisely, we want to know what properties should the CPTP map $\Phi$ satisfy in order to have
\begin{align}\label{eq:cond-for-eq}
\ln[\Phi^\dag(\Gamma^S_\tau)] =\Phi^\dag(\ln\Gamma^S_\tau)\;.
\end{align}
Clearly, for unitary $\Phi$ equality is obtained. However, unitarity is not necessary: for example, any CPTP map of the form
\begin{align}\label{eq:counter-unital}
\Phi(\cdot^S)=\sum_{i=1}|f_i\>\<\varphi_i|^S(\cdot^S)|\varphi_i\>\<f_i|^S\;,
\end{align}
where $\{|f_i\>^S \}_i$ is the energy eigenbasis for $H^S_\tau$ and $\{|\varphi_i\>^S \}_i$ is an arbitrary orthonormal basis for $S$, also attains~(\ref{eq:cond-for-eq}). Maps like that in~(\ref{eq:counter-unital}), even though not unitary, still satisfy the relation $\Phi(\openone)=\openone$. For such maps, called unital, also their trace-dual $\Phi^\dag$ is CPTP. More freedom in constructing counterexamples is available whenever $H^S_\tau$ has some degenerate eigenvalues: to see this, it suffices to consider the case of a completely degenerate $H^S_\tau$ for which $\Gamma^S_\tau\propto\openone^S$, so that condition~(\ref{eq:cond-for-eq}) is automatically satisfied because trace-duals of CPTP maps always preserve the identity operator.

\section{Thermodynamic length}

One of the advantages of having a bound expressed in terms of a relative entropy is that it is possible to provide such a bound with a geometrical interpretation~\cite{amari_2016}.

Following~\cite{DeffnerLutz2010}, let us consider a closed system, initially in thermal equilibrium ($\rho^S_0=\Gamma^S_0$), driven by a thermodynamic protocol that brings the system's Hamiltonian from $H_0^S$ to $H_\tau^S$ by changing the values of some controllable macroscopic variables. We follow here the common semiclassical assumption of ``ideal pistons,'' according to which any back-reactions of the driven system on the thermodynamic drive can be neglected. Even in the presence of a time-varying Hamiltonian, the evolution from $t=0$ to $t=\tau$ of a closed system remains unitary, so that both $\Phi$ and $\Phi^\dag$ are unitary maps. This fact in particular implies that the operator $L$ in Eq.~(\ref{eq:main-eq}) is null. We thus obtain the following:
\begin{align}
\beta(\Delta E-\Delta F)&=D(\Gamma^S_0\|\Phi^\dag(\Gamma^S_\tau))\label{eq:closed-cyclic}\\
&=D(\Phi(\Gamma^S_0) \|\Gamma^S_\tau)\label{eq:deffner-lutz}\\
&\ge 0\;,\nonumber
\end{align}
where the second line follows again from the unitarity of $\Phi$. (Alternatively, it can be derived from~(\ref{eq:standard-eq}) using the fact that for a unitary transformation $\Delta S=0$.)

Relation~(\ref{eq:deffner-lutz}) has been first derived in~\cite{DeffnerLutz2010}, where it is interpreted as an expression of Clausius inequality for the irreversible entropy production in a nonequilibrium quantum process. Then, one of the main results of~\cite{DeffnerLutz2010} is to apply a bound of~\cite{AudenaertEisert2005} to show that
\begin{align}\label{eq:deffner-lutz-bures}
\beta(\Delta E-\Delta F)\ge \frac{8}{\pi^2}\mathcal{L}^2(\Phi(\Gamma^S_0),\Gamma^S_\tau)\;,
\end{align}
where $\mathcal{L}(\rho_1,\rho_2):=\arccos \fid(\rho_1,\rho_2)$, in turn defined in terms of the fidelity $\fid(\rho_1,\rho_2):=\N{\sqrt{\rho_1}\sqrt{\rho_2}}_1$, is a Riemannian metric called the \textit{Bures angle} or \textit{Bures length} between the normalized density matrices $\rho_1$ and $\rho_2$~\cite{bengtsson_zyczkowski_2006}. Relation~(\ref{eq:deffner-lutz-bures}) is interpreted in~\cite{DeffnerLutz2010} as a quantum generalization of the thermodynamic length~\cite{Crooks-thermo-length-2007}.

The thermodynamic reverse bound~(\ref{eq:main-ineq}) allows us to obtain a similar conclusion for open processes as well. We proceed as follows. First of all, we introduce a normalized density matrix $\gamma^S:=\frac{1}{\Tr{\Phi^\dag(\Gamma^S_\tau)}}\Phi^\dag(\Gamma^S_\tau)$, in terms of which we can write
\begin{align}
&D(\Gamma^S_0\|\Phi^\dag(\Gamma^S_\tau))\nonumber\\
&=D(\Gamma^S_0\|\gamma^S)-\ln\Tr{\Phi^\dag(\Gamma^S_\tau)}\;.\label{eq:ineq-renorm}
\end{align}
Then, plugging the above relation into~(\ref{eq:main-ineq}) and applying the same bound as in~\cite{DeffnerLutz2010}, we obtain
\begin{align}
&\beta(\Delta E-\Delta F)\nonumber\\
&\ge \frac{8}{\pi^2}\mathcal{L}^2(\Gamma^S_0,\gamma^S)-\ln\Tr{\Phi^\dag(\Gamma^S_\tau)}\;.
\label{eq:deffner-lutz-open}
\end{align}
The above relation is rigorously valid for any CPTP map $\Phi$ and reduces to relation~(\ref{eq:deffner-lutz-bures}) for unitary processes. In general, however, while the above bound always holds, relation~(\ref{eq:deffner-lutz-bures}) only holds for closed processes.

In order to illustrate this point further, let us consider the case in which the trace-dual map $\Phi^\dag$ preserves the trace of the final equilibrium state, that is, $\Tr{\Phi^\dag(\Gamma^S_\tau)}=1$. In this case, the extra term in relation~(\ref{eq:deffner-lutz-open}) disappears. Relation~(\ref{eq:deffner-lutz-open}) then assumes the same form of relation~(\ref{eq:deffner-lutz-bures}), which is however derived under the assumption that the process is closed. Indeed, in general, $\mathcal{L}^2(\Gamma^S_0,\gamma^S)\neq \mathcal{L}^2(\Phi(\Gamma^S_0),\Gamma^S_\tau)$, and while Eq.~(\ref{eq:deffner-lutz-open}) always holds, Eq.~(\ref{eq:deffner-lutz-bures}) may be violated in open processes. Hence, relation~(\ref{eq:deffner-lutz-open}) should be preferred over~(\ref{eq:deffner-lutz-bures}). In other words, the distance between the initial equilibrium state $\Gamma^S_0$ and the ``retrodicted'' state $\gamma^S$ appears to be a better indicator than the distance between the evolved initial equilibrium state and the final one.

An important family of open processes for which $\Tr{\Phi^\dag(\Gamma^S_\tau)}=1$ is the family of CPTP unital maps, namely, CPTP maps such that also their trace-dual $\Phi^\dag$ is TP. We thus see that a direct connection between the energy cost and the Bures length between $\Gamma^S_0$ and $\gamma^S$ exists not only in closed processes, but also in open unital processes. This fact can be seen as yet another instance of the general intuition that unital open processes, with respect to various thermodynamic aspects, behave ``almost as'' closed ones~\cite{Rastegin_2013,Albash2013}. Moreover, for unital $\Phi$, the trace-dual $\Phi^\dag$ can be interpreted as the backward process~\cite{Albash2013} and the state $\gamma^S:=\Phi^\dag(\Gamma^S_\tau)$ gains a direct physical interpretation.

We conclude this section by noticing that an alternative relation with the fidelity can be derived, which is tighter than~(\ref{eq:deffner-lutz-open}), even though it loses the character of being expressed in terms of a metric. Such relation, which is a direct consequence of a bound in~\cite{MDSFT2013} (see in particular Theorems~5 and~7 therein), reads as follows:
\begin{align}
&\beta(\Delta E-\Delta F)\nonumber\\
&\ge -2\ln \fid(\Gamma^S_0,\gamma^S)-\ln\Tr{\Phi^\dag(\Gamma^S_\tau)}\;.\label{eq:in-terms-of-fid}
\end{align}

\section{Erasure processes}

A paradigmatic example of open system evolution is provided by erasure processes. Following~\cite{GooldModiPater2015}, we define an erasure process as any process of the form
\begin{align}\label{eq:erasure}
\mE(\Gamma^A_0)=\PTr{S}{U(\rho^S_0\otimes\Gamma^A_0)U^\dag}\;,
\end{align}
where $U$ represents an arbitrary, though fixed, interaction mechanism between the system (whose state is being ``erased'') and the ancilla, which plays here the role of the environment. The resulting CPTP map $\mE$ is a map that acts on the ancilla only, its action depending on the system's initial state. For the sake of concreteness, in Eq.~(\ref{eq:erasure}) the ancilla is assumed to be initialized in the equilibrium state $\Gamma^A_0$, and the ancilla's Hamiltonian $H^A$ is assumed to remain constant during the process, so that $\Gamma^A_0=\Gamma^A_\tau=:\Gamma^A$. In this situation, the change in the ancilla's internal energy $\Delta E=\Tr{\mE(\Gamma^A)\ H^A}-\Tr{\Gamma^A\ H^A}$ is interpreted as the average heat $\<Q\>$ that flows from the system into the environment during the erasure process.

In~\cite{GooldModiPater2015} the following bound on the average heat flow is derived:
\begin{align}\label{eq:erasure-bound-GMP}
\beta\<Q\>\ge-\ln\Tr{\mE^\dag(\Gamma^A)}\;.
\end{align}
The quantity appearing in the right-hand side of the above inequality coincides exactly with the extra term we obtained by renormalizing the relative entropy in~(\ref{eq:ineq-renorm}). Therefore, we know that inequality~(\ref{eq:erasure-bound-GMP}) can be made tighter by adding the term $D(\Gamma^A\|\gamma^A)$, which measure the entropic divergence between the equilibrium state $\Gamma^A$ and the renormalized retrodicted state $\gamma^A:=\frac{1}{\Tr{\mE^\dag(\Gamma^A)}}\mE^\dag(\Gamma^A)$. We thus obtain
\begin{align}\label{eq:erasure-bound-ours}
\beta\<Q\>&\ge D(\Gamma^A\|\gamma^A)-\ln\Tr{\mE^\dag(\Gamma^A)}\;,
\end{align}
which, in comparison with the analogous bound in~\cite{GooldModiPater2015}, has two advantages: it is stronger, because $D(\Gamma^A\|\gamma^A)\ge 0$, and it possesses a geometric interpretation in terms of the thermodynamic length, as discussed above.

\subsection{Comparison}

From the properties of the quantum relative entropy, relation~(\ref{eq:erasure-bound-ours}) above becomes \textit{strictly} stronger than~(\ref{eq:erasure-bound-GMP}), as long as $D(\Gamma^A\|\gamma^A)>0$, which is the case whenever $\gamma^A\neq\Gamma^A$, that is, $\mE^\dag(\Gamma^A)\centernot{\propto}\Gamma^A$. On the other hand, the extra term $D(\Gamma^A\|\gamma^A)$ cannot become infinite, since $\mE^\dag(\Gamma^A)$ is invertible whenever $\Gamma^A$ is. This is a consequence of the fact that, for any CPTP map $\Phi$, its trace-dual map $\Phi^\dag$ is spectrum-width decreasing (see, e.g., Lemma~3.1 in~\cite{cleanPOVMs}).

The gap between~(\ref{eq:erasure-bound-ours}) and~(\ref{eq:erasure-bound-GMP}), though not infinite, can however be arbitrarily large, as the following example shows. Let us consider the case in which the erasure map has the form $\mE(\rho^A)=\sigma^A$, independently of the initial state $\rho^A$ of the ancilla. This corresponds to the situation in which erasure of the system is achieved by simply ``dumping'' it into the environment. In this case, the trace-dual map is easily computed as
\[
\mE^\dag(\cdot)=\openone^A\Tr{\cdot\ \sigma^A}\;.
\]
Then, by direct inspection we find that, in this particular case,
\[
D(\Gamma^A\|\gamma^A)=\ln d_A-S(\Gamma^A),
\]
being $d_A$ the dimension of the ancilla. Hence, by decreasing the temperature of the ancilla, the term $D(\Gamma^A\|\gamma^A)$ can go up to $\ln d_A$ (assuming nondegenerate $H^A$). In Figure~1 below we report some plots in which we numerically compare Eqs.~\eqref{eq:erasure-bound-GMP} and~\eqref{eq:erasure-bound-ours} in the case of a spin-1/2 system coupled to an interacting spin chain at finite temperature.

\section{Maximal energy outflow}

Until now we have used relation~(\ref{eq:main-ineq}) in order to bound from below the net amount of {energy that flowed} into the system during the process modeled by the CPTP map $\Phi$. Of course it can also be used to provide upper bounds on the maximum amount $\Delta_{\operatorname{drop}}=E_0-E_\tau$ by which the system's internal energy can \textit{drop} as a consequence of $\Phi$. Again assuming for simplicity that $H^S_0=H^S_\tau=:H^S$ and $\rho^S_0=\Gamma^S_0=:\Gamma^S$, the thermodynamic reverse bound~(\ref{eq:main-ineq}) together with~(\ref{eq:ineq-renorm}) immediately provide an upper bound to $\Delta_{\operatorname{drop}}$ as 
\begin{align}
\Delta_{\operatorname{drop}}&\le-\beta^{-1}D(\Gamma^S\|\Phi^\dag(\Gamma^S))\nonumber\\
&=\beta^{-1}\{\ln\Tr{\Phi^\dag(\Gamma^S)}-D(\Gamma^S\|\gamma^S)\}\nonumber\\
&\le\beta^{-1}\ln\Tr{\Phi^\dag(\Gamma^S)}\;.\label{eq:our-drop}
\end{align}
If the CPTP map $\Phi$ is obtained from an interaction model as in Eq.~(\ref{eq:micro-model}), it is possible to show, as done in Appendix~\ref{app:A}, that, for any state $\sigma$, $\Tr{\Phi^\dag(\sigma)}\le d_A$. In this way we recover Theorem~1 of~\cite{Tasaki2015}, that is
\begin{align}\label{eq:tasaki}
\Delta_{\operatorname{drop}}\le\beta^{-1}\ln d_A\;.
\end{align}
The above relation was interpreted in~\cite{Tasaki2015} as an extension of Lenard's result~\cite{Lenard1978} to general CPTP processes. Notice however that, while relation~(\ref{eq:tasaki}) depends on the particular microscopic model chosen as the realization of the process $\Phi$, our bound, Eq.~(\ref{eq:our-drop}), only depends on the system's reduced dynamics $\Phi$.

\section{Conclusions}

In this paper we discussed what we named as thermodynamic reverse bound, a bound for general open quantum processes in terms of a trace-dual map and quantum relative entropy. We showed that this bound improves some known thermodynamic inequalities. This bound invites further research, not only because of its usability to improve thermodynamic inequalities but also for its relevance to the study of reverse processes.

We conclude this paper noticing that the thermodynamic reverse bound~(\ref{eq:main-ineq}) and its applications---that is, inequalities~(\ref{eq:deffner-lutz-open}), (\ref{eq:in-terms-of-fid}), (\ref{eq:erasure-bound-ours}), and~(\ref{eq:our-drop})---also hold under the weaker assumption that the map $\Phi$ is trace-preserving and positive, not necessarily \textit{completely} positive~\cite{Choi1974}. In such a case, however, the map $\Phi$ does not admit in general a microscopic model such as in Eq.~(\ref{eq:micro-model}) and its interpretation remains unclear~\cite{Pechukas-1,Alicki-comment,Pechukas-replies,Buscemi-CP}. 

\begin{acknowledgments}
This work was supported by MEXT Quantum Leap Flagship Program (MEXT Q-LEAP), Grant Number JPMXS0120319794, and by the Japan Society for the Promotion of Science (JSPS) KAKENHI, Grants Number 19H04066 and Number 20K03746.
\end{acknowledgments}

\onecolumngrid

\begin{figure}\label{fig:grafici}
	\centering
	\subfigure[]{\includegraphics[width=0.32\linewidth]{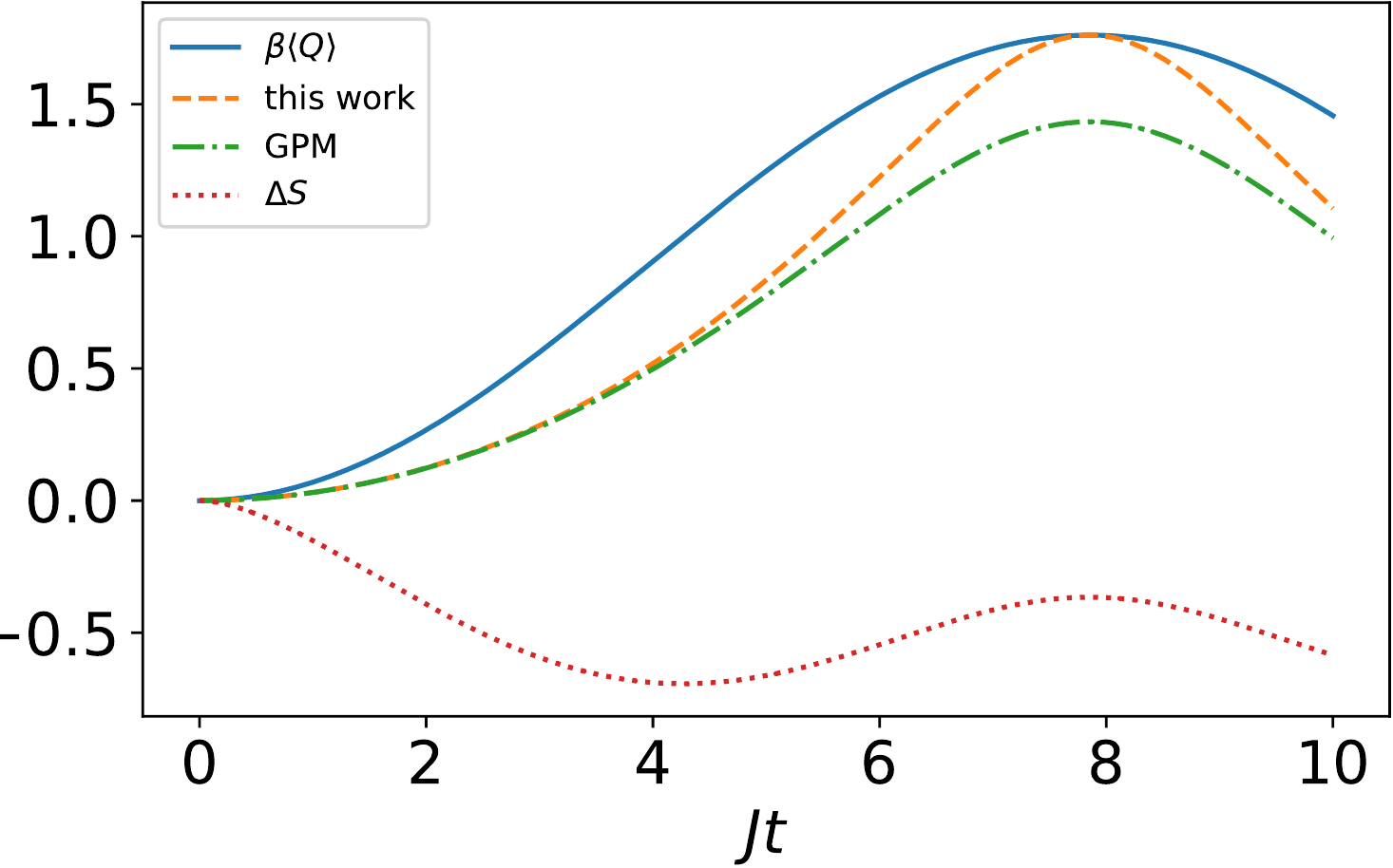}}
	\subfigure[]{\includegraphics[width=0.32\linewidth]{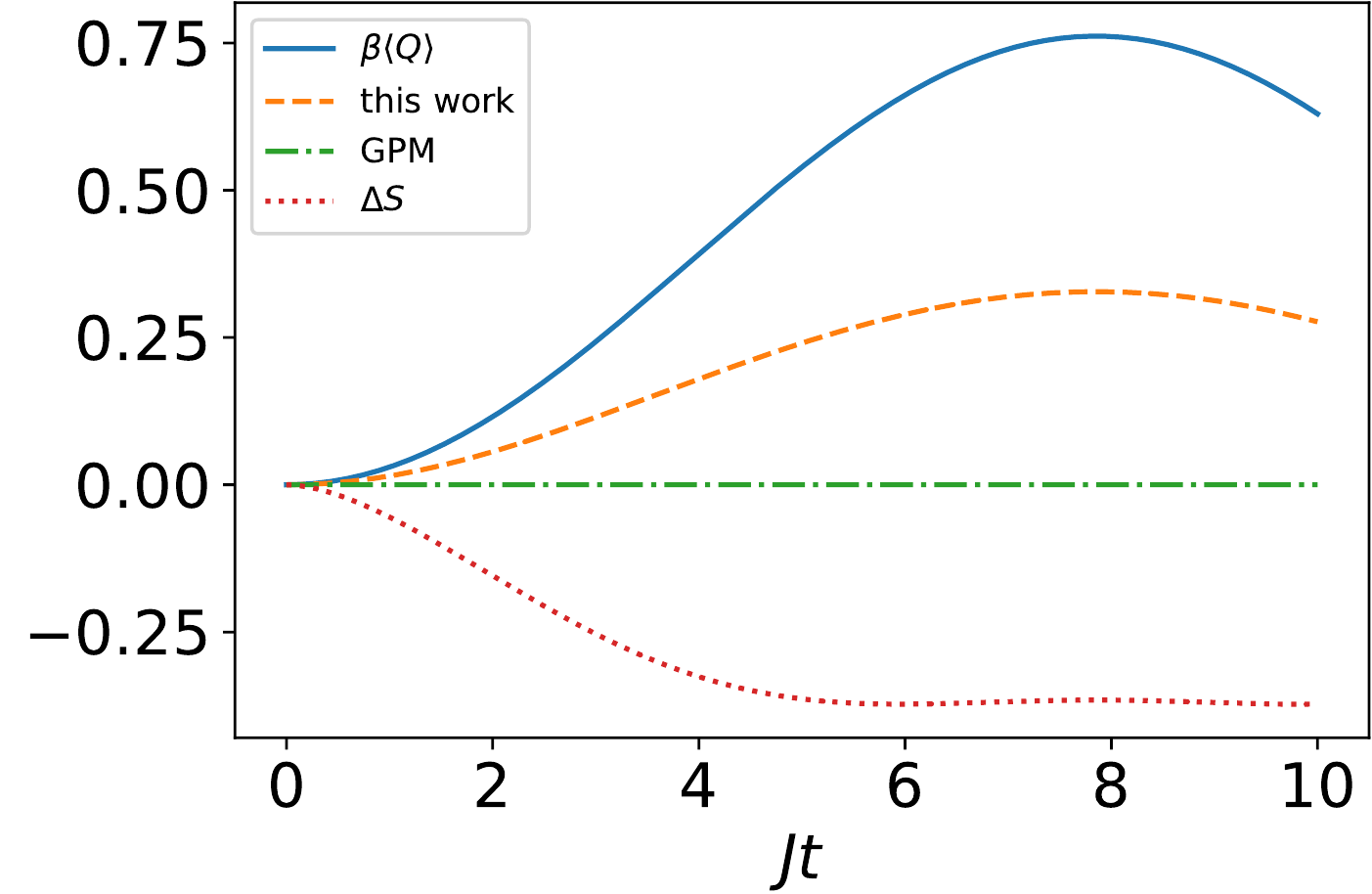}}	
	\subfigure[]{\includegraphics[width=0.32\linewidth]{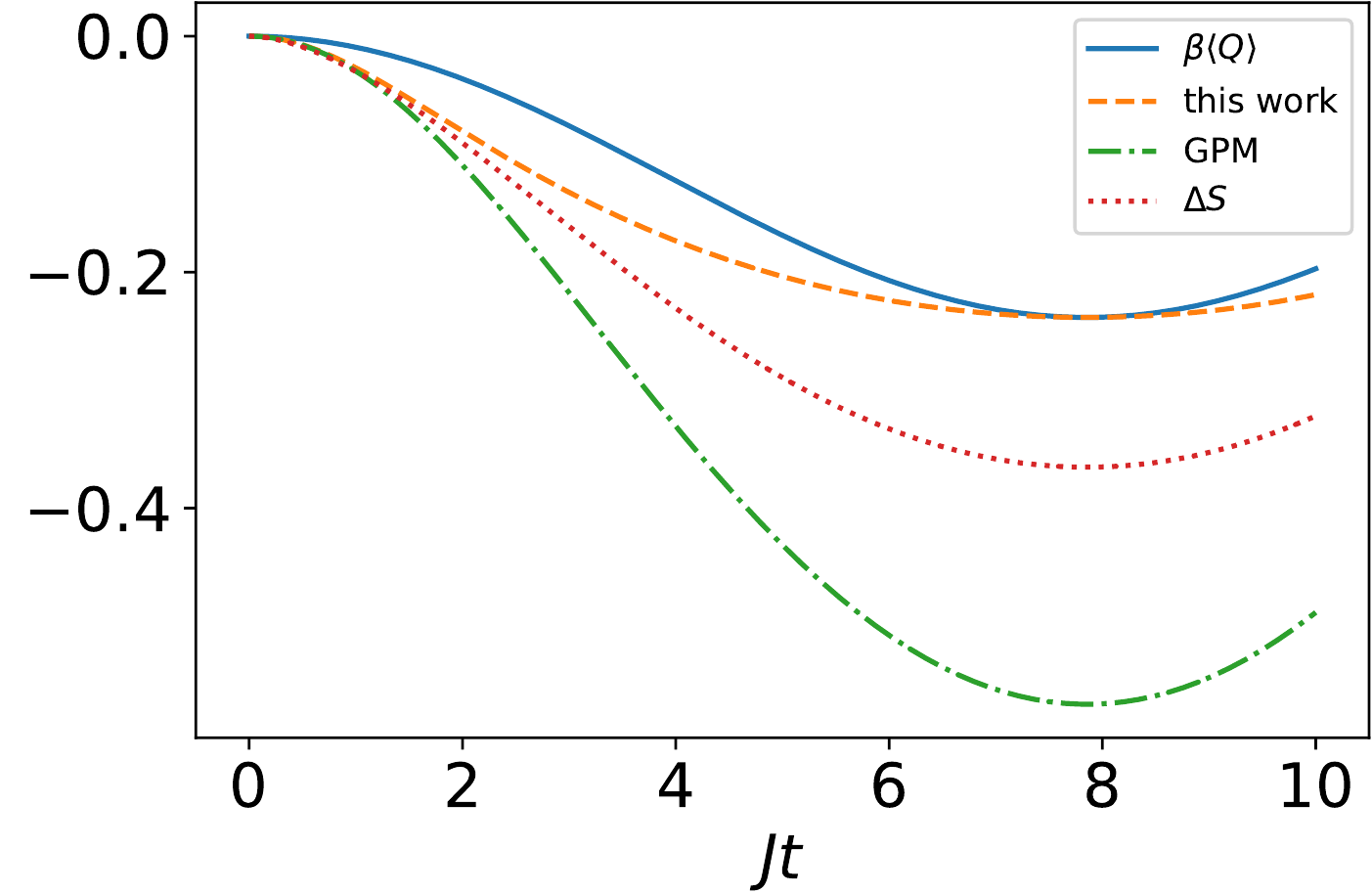}}
	\subfigure[]{\includegraphics[width=0.32\linewidth]{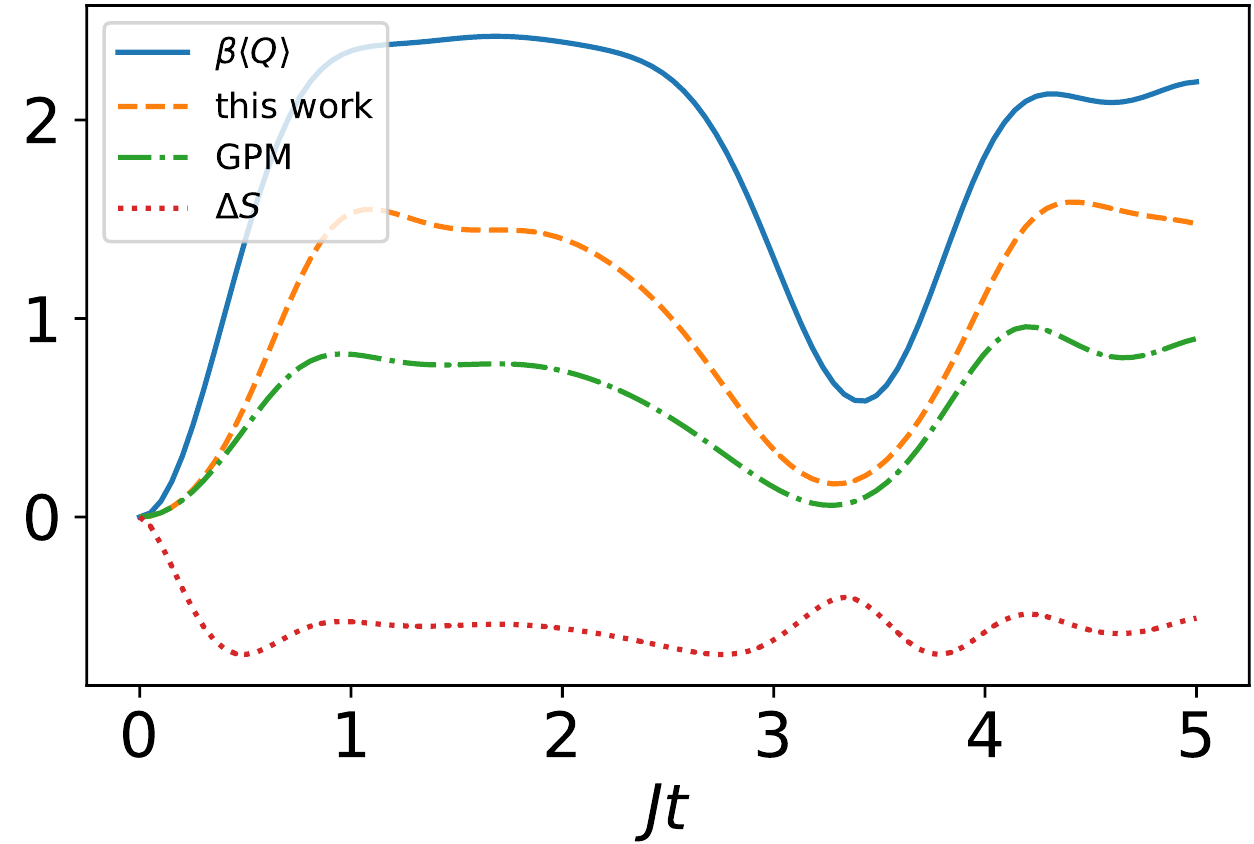}}
	\subfigure[]{\includegraphics[width=0.32\linewidth]{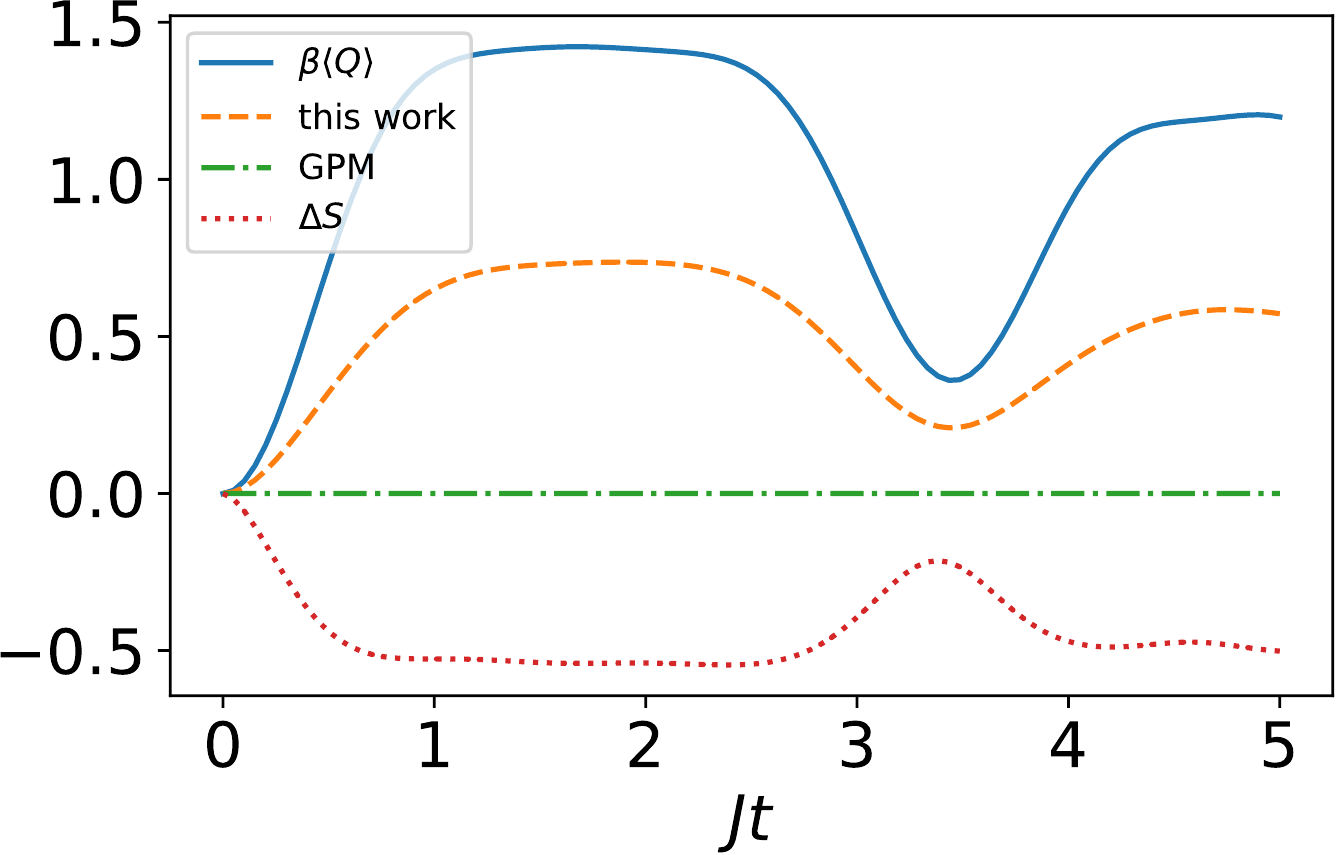}}
	\subfigure[]{\includegraphics[width=0.32\linewidth]{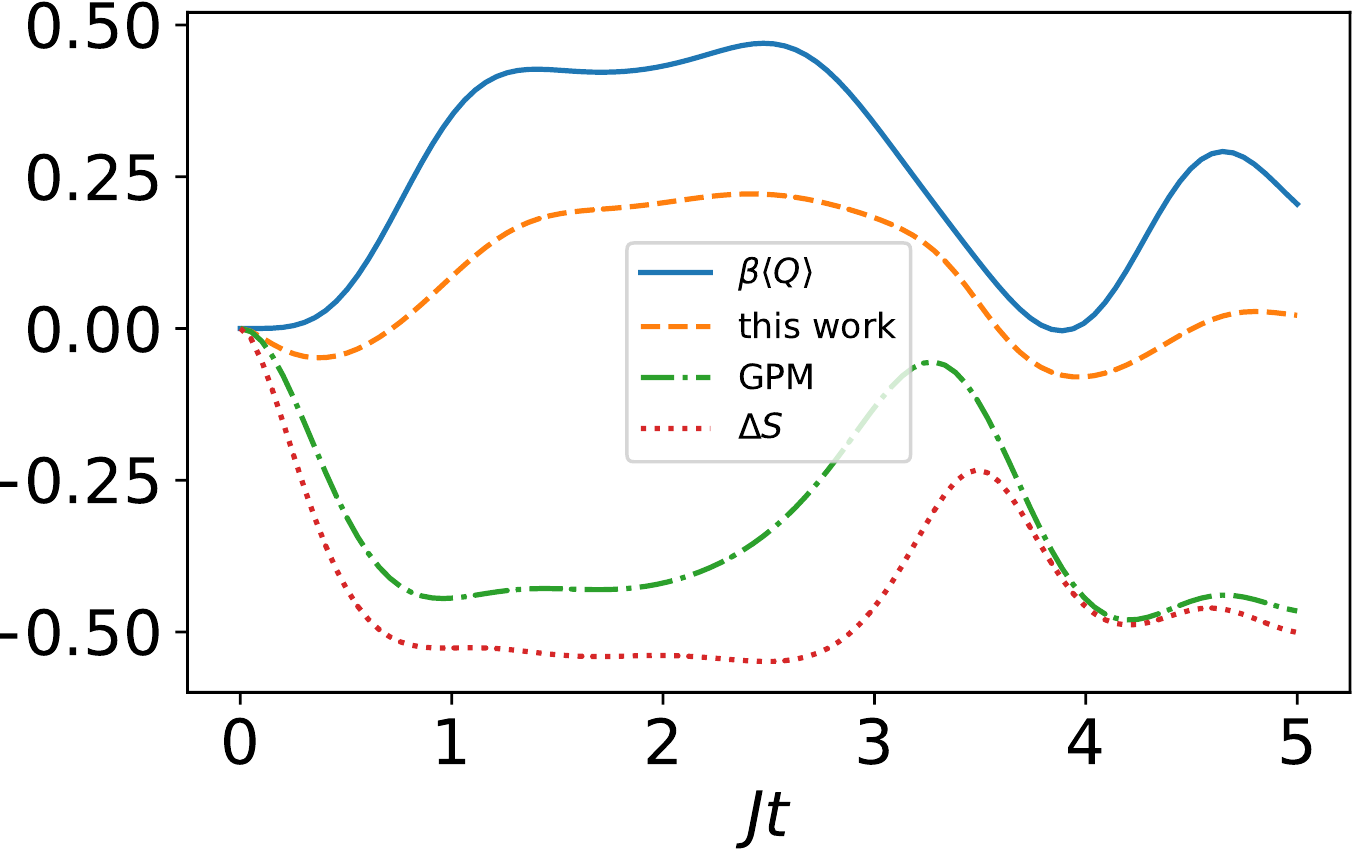}}
	\caption{The six plots above show the explicit values of the quantities appearing in Eq.~\eqref{eq:erasure-bound-ours}, in the case of a spin-1/2 system interacting with a spin-1/2 ancilla (panels (a), (b), and (c)), and in the case of a spin-1/2 system interacting with a four spin-1/2-particle ancilla (panels (d), (e), and (f)). In all panels the continuous curve corresponds to $\beta\langle Q\rangle$, the dashed curve corresponds to the right-hand side of~(\ref{eq:erasure-bound-ours}), the dashdotted curve corresponds to the right-hand side of~(\ref{eq:erasure-bound-GMP}), the dotted curve corresponds to $\Delta S$ (expressed in nats). \textbf{Interaction model.} In order to facilitate the comparison with Ref.~\cite{GooldModiPater2015}, we use exactly the same microscopic model: the Hamiltonian of the ancilla is assumed to be $H^A=J\sum_{j=1}^{N-1}(\sigma_x^j\sigma_x^{j+1}+\sigma_y^j\sigma_y^{j+1})+B\sum_{j=1}^{N}\sigma_z^j$, where $N$ is the number of spin-1/2 particles constituting the ancillary degrees of freedom, the $\sigma$'s denote the Pauli matrices, $J$ is the interspin coupling strength, and $B$ is the coupling of each ancillary spin with a homogeneous external magnetic field. The spin-1/2 system is coupled to the ancilla only through the first ancillary spin as follows: $H^{SA}=J_0(\sigma_x^S\sigma_x^{1}+\sigma_y^S\sigma_y^{1})$. Finally, the system's Hamiltonian is $H^S=B\sigma^S_z$. Following~\cite{GooldModiPater2015}, for the numerics we put $J_0=0.1$ in panels (a--c), $J_0=1$ in panels (d--f), while $B=J=\beta=1$ everywhere. Both axes report adimensional quantities, and the scale on the $x$-axis measures $Jt$, where $t$ is time. \textbf{Initial conditions.} The ancilla is assumed to always start in its thermal equilibrium state. The system instead is assumed to begin in a pure state $\alpha|\uparrow_z\>+\sqrt{1-\alpha^2}|\downarrow_z\>$. In particular: in plots (a) and (d) we set $\alpha=1$; in plots (b) and (e), $\alpha=1/\sqrt{2}$; in plots (c) and (f), $\alpha=0$. Hence, panels (a) and (d), respectively, reproduce Fig.~1(b) and Fig.2 of~\cite{GooldModiPater2015}, respectively. \textbf{Comparison.} From the plots we notice that our bound (the dashed curve) is always the closest one to the curve corresponding to the actual heat exchange (continuous curve), even when the other known bounds performs poorly as in panels (b), (c), (d), and (e).}
\end{figure}

\twocolumngrid

\bibliographystyle{alphaurl}
\bibliography{passive-channels-bib}

\appendix

\section{Appendix: derivation of bound~(\ref{eq:tasaki}) in the main text}\label{app:A}

\begin{lemma}
	Consider a CPTP map $\Phi$ defined by the equation
	\begin{equation}\label{eq:gen-CPTP}
	\Phi(\cdot^S):=\PTr{\mathcal{A}}{U^{SA}(\cdot^S\otimes\rho^A)(U^{SA})^\dag}\;,
	\end{equation}
	and denote by $\Phi^\dag$ the trace-dual of $\mE$, that is, the map satisfying the relation $\Tr{\Phi(X)\ Y}=\Tr{X\ \Phi^\dag(Y)}$, for all operators $X,Y$. Then, for any normalized density matrix $\sigma$,
	\begin{equation}
	\Tr{\Phi^\dag(\sigma)}\le d_A\;.
	\end{equation}
\end{lemma}

\begin{proof}
	Let us fix an arbitrary orthonormal basis in the appartus' Hilbert space, let us say, $\{|\varphi_i\>^A \}$, and let us begin by assuming that the apparatus' initial state is pure, let us say, $\rho^A=|\psi\>\<\psi|^A$. In this case, the CPTP map $\Phi$ can be written as
	\[
	\Phi(\cdot)=\sum_{i=1}^NK_i\cdot K_i^\dag\;,
	\]
	where the $N=d_A$ Kraus operators $K_i$ are defined by the relation
	\[
	K_i:=(\openone^S\otimes\<\varphi_i|^A)\ U^{SA}\ (\openone^S\otimes|\psi\>^A)\;.
	\]
	Then,
	\begin{align}
	\Tr{\Phi^\dag(\sigma)}&=\Tr{\sum_{i=1}^N K_iK_i^\dag\ \sigma}\\
	&\le \N{\sum_{i=1}^N K_iK_i^\dag}_\infty\label{eq:lmax}\\
	&\le \sum_{i=1}^N\N{K_iK_i^\dag}_\infty\label{eq:triangle}\\
	&\le\sum_{i=1}^N 1\label{eq:polar}\\
	&=N=d_A\;,
	\end{align}
	where
	\begin{itemize}
		\item in~(\ref{eq:lmax}) we used the inequality $\Tr{X\sigma}\le\lambda_{\max}(X)=:\N{X}_\infty$, valid for all self-adjoint operators $X$ and all normalized density matrices $\sigma$;
		\item in~(\ref{eq:triangle}) we used the triangle inequality for the Schatten infinity norm
		\item in~(\ref{eq:polar}) we used the fact that, for all $i$, $K_iK_i^\dag$ and $K_i^\dag K_i$ have the same eigenvalues (a consequence of their being positive semidefinite and of the polar decomposition), jointly with the fact that $\sum_iK_i^\dag K_i=\openone$, so that $K_iK_i^\dag\le\openone$ for all $i$.
	\end{itemize}
	
	Let us now consider the general case of a mixed initial apparatus's state $\rho^A$. In this case, the density matrix can be expanded on its pure components $\rho^A=\sum_jp(j)|\chi_j\>\<\chi_j|^A$, where $p(j)$ is a probability distribution. By plugging this into Eq.~(\ref{eq:gen-CPTP}) and using the linearity of the partial trace, we see that $\Phi(\cdot)=\sum_jp(j)\Phi_j(\cdot)$, where the $\Phi_j$'s are all CPTP maps as the original $\Phi$, but realized by means of a \textit{pure} apparatus state. Therefore, by following the same lines of reasoning as before, each $\Phi_j$ can be written using only $N=d_A$ Kraus operators. Thus,
	\begin{align}
	\Tr{\Phi^\dag(\sigma)}&=\Tr{\left(\sum_{j}p(j) \Phi_j^\dag\right)(\sigma)}\\
	&=\sum_jp(j)\Tr{\Phi^\dag_j(\sigma)}\\
	&\le \sum_jp(j)\sum_{i=1}^N 1\\
	&=N=d_A\;.
	\end{align}
\end{proof}

\end{document}